\documentclass[10pt,amssymb,amsfont,a4paper]{article}
\usepackage{latexsym,amssymb,amsmath,amsthm,color}
\usepackage{geometry}
\usepackage{url}
\usepackage{graphicx}
\usepackage[utf8]{inputenc}
\usepackage[numbers]{natbib}
\usepackage{authblk}
\usepackage{tikz}
\usetikzlibrary{decorations.pathreplacing}
\usetikzlibrary{shapes.misc, fit}

\DeclareMathOperator{\diag}{Diag}

\DeclareMathOperator{\dd}{d}

\theoremstyle{definition}
\newtheorem{definition}{Definition}
\newtheorem{construction}{Construction}

\theoremstyle{plain}
\newtheorem{theorem}{Theorem}

\newtheorem{lemma}[definition]{Lemma}

\newtheorem{corollary}[definition]{Corollary}

\title{Private Information Retrieval from Locally Repairable Databases with Colluding Servers}
\author{Umberto Mart{\'i}nez-Pe\~{n}as \thanks{umberto.martinez@uva.es}}
\affil{IMUVa-Mathematics Research Institute\\University of Valladolid, Spain}

\date{}

\begin{document}

\maketitle

\begin{abstract}
We consider information-theoretical private information retrieval (PIR) from a coded database with colluding servers. We target, for the first time, locally repairable storage codes (LRCs). We consider any number of local groups $ g $, locality $ r $, local distance $ \delta $ and dimension $ k $. Our main contribution is a PIR scheme for maximally recoverable (MR) LRCs based on linearized Reed--Solomon codes, which achieve the smallest field sizes among MR-LRCs for many parameter regimes. In our scheme, nodes are identified with codeword symbols and servers are identified with local groups of nodes. Only locally non-redundant information is downloaded from each server, that is, only $ r $ nodes (out of $ r+\delta-1 $) are downloaded per server.  The PIR scheme achieves the (download) rate $ R = (N - k - rt + 1)/N $, where $ N = gr $ is the length of the MDS code obtained after removing the local parities, and for any $ t $ colluding servers such that $ k + rt \leq N $. For an unbounded number of stored files, the obtained rate is strictly larger than those of known PIR schemes that work for any MDS code. Finally, the obtained PIR scheme can also be adapted when communication between the user and each server is performed via linear network coding, achieving the same rate as previous PIR schemes for this scenario but with polynomial finite field sizes, instead of exponential. Our rates are equal to those of PIR schemes for Reed--Solomon codes, but Reed--Solomon codes are incompatible with the MR-LRC property or linear network coding, thus our PIR scheme is less restrictive in its applications. 

\end{abstract}

\section{Introduction} \label{sec intro}

Private information retrieval (PIR), introduced in \cite{chor-focs, chor-acm}, consists in retrieving a file from a database without revealing the index of the retrieved file to the servers, hence keeping the user's preference for a file private. In this work, we consider information-theoretical privacy, meaning that an undesired observer (e.g., the servers) with unbounded computational power may not obtain any information on the file index. Originally \cite{chor-focs, chor-acm}, databases were considered to store files using a \textit{repetition} code, that is, each server stores one copy of each file, and servers were not considered to communicate with each other (\textit{collude}) in order to gain information on the file index. 

As it was pointed out in these seminal works, an obvious solution to the PIR problem is to download the entire database. However, this turns out to be wasteful and much higher \textit{download rates}, or simply \textit{rates} (the size of the file divided by the amount of downloaded data), can be achieved when more than one server is used. It was also shown in these seminal works that downloading the whole database is the only solution in the single-server case (for information-theoretical privacy). 

As is well-known, databases often suffer from data erasures (due to disk failures). Using a repetition code, that is, storing copies of the same files across multiple servers, yields an unacceptably high overhead (i.e., unacceptably low information rate). PIR from a database where data is stored after being encoded by a non-repetition code (\textit{coded database}) was considered in \cite{blackburn, fazeli, extrabit}. However, the number of servers in these works is either larger than the number of files or grows as the overhead of the storage code decreases, which are not practical scenarios.

Explicit PIR schemes from a database that uses an MDS storage code, and for $ \tau \geq 1 $ colluding servers, were obtained in \cite{hollanti, razan-it, razan-isit, zhang-ge}. The optimal rate for a fixed number $ m $ of stored files is
\begin{equation}
R = \frac{N - k - \tau + 1}{N} \cdot \left( 1 - \left( \frac{k + \tau - 1}{N} \right)^{m} \right)^{-1},
\label{eq highest rates}
\end{equation}
for the cases $ \tau = 1 $ \cite{banawan-ulukus, sun-jafar} and $ k = 1 $ \cite{sun-jafar-colluding}, although (\ref{eq highest rates}) is not optimal in general \cite{conjecture}. If the MDS storage code has dimension $ k $ and length $ N $, \textit{universal} PIR schemes (compatible with any $ (N,k) $ MDS storage code) were obtained in \cite{razan-it, razan-isit} with rate $ R = 1/ N $, for $ \tau = N-k $, and rate $ R = (N-k)/N $ for $ \tau = 1 $. For $ \tau = 1 $ or $ k = 1 $, these rates tend to the optimal rate (\ref{eq highest rates}) as the number of files $ m $ increases, and in fact they become very close already for moderate values of $ m $ (since $ m $ appears in the exponent). We also remark that, in practical scenarios, we usually have $ m \gg N $.

For $ \tau > 1 $ and $ k > 1 $, the optimal rate is unknown in general, but the rate of the PIR scheme in \cite{hollanti}, 
\begin{equation} \label{eq camilla rate}
R = \frac{N - k - \tau + 1}{N} ,
\end{equation}
is the highest known (and it is unmatched by other schemes) for an unrestricted number of files and is always strictly larger than that of known universal PIR schemes \cite{razan-it}. 

However, we will disregard in our rate comparisons the scheme in \cite{hollanti}, since it is \textit{more restrictive} as it only works for generalized Reed--Solomon (GRS) storage codes \cite{reed-solomon}, but not for other MDS storage codes, which makes them incompatible with any MR-LRC coded database (except for the trivial cases). More concretely, the scheme in \cite{hollanti} uses the coordinate-wise product, and the only MDS codes that satisfy the desired properties with respect to such products are GRS codes \cite{marquez}. Furthermore, GRS codes are not LRCs and can never be (except for trivial cases) the MDS codes obtained after removing the local parities of an MR-LRC. This is because GRS codes have linear field sizes and MR-LRCs require super-linear field sizes \cite{gopi}. Hence the scheme from \cite{hollanti} may not be adapted to be used in an MR-LRC coded database. For similar reasons, the PIR scheme from \cite{hollanti} is also incompatible with linear network coding (see Subsection \ref{subsec pir over lin coded networks} and \cite{pir-networks}). 

As pointed out in the distributed storage literature, MDS codes are not suitable for large databases, which are increasingly more important due to the spread of Big Data. This is because repairing one single failed node with an MDS code requires contacting and downloading the content of a large number of nodes, resulting in a high repair latency. Locally repairable codes (LRCs), introduced in \cite{gopalan, kamath} and already applied in practice (by, e.g., Microsoft \cite{azure} and Facebook \cite{xoring}), allow to repair a single erasure (or generally $ \delta - 1 $ erasures per \textit{local group}, for a \textit{local distance} $ \delta $) by contacting at most a number $ r $, called \textit{locality}, of other nodes. Maximally recoverable (MR) LRCs were introduced in \cite{blaum-RAID, gopalan-MR} and are optimal in the following strong sense: Given parameters $ k $, $ r $, $ \delta $ and number of local groups $ g $, if there is an erasure pattern that an MR-LRC with such parameters cannot correct, then such a pattern cannot be corrected by any other LRC with such parameters. Such patterns can be described easily, see Definitions \ref{def LRC} and \ref{def MR}. 

To the best of our knowledge, no work has provided PIR schemes with rates as high as (\ref{eq camilla rate}) for optimal LRCs, MR-LRCs or MDS codes obtained from puncturing MR-LRCs for general parameters $ g $, $ r $, $ \delta $, $ k $, $ \tau $ and $ N = gr $. As explained above, the only PIR schemes for some MDS coded databases with PIR rates as large as (\ref{eq camilla rate}) are those from \cite{hollanti}. However, these PIR schemes only work for GRS codes, since they make use of coordinate-wise products and GRS codes are the only MDS codes suitable for such products \cite{marquez}. Moreover, GRS codes can never be the MDS codes obtained from puncturing MR-LRCs (except for trivial cases) since they have linear field sizes and MR-LRCs require super-linear field sizes \cite{gopi}. We will circumvent the limitations of \cite{hollanti} by making use of coordinate-wise matrix products (see Section \ref{sec coordinate matrix product}). 

In this work, we provide the first PIR scheme for a class of MR-LRC storage codes that cover general parameters. We consider the MR-LRCs from \cite{universal-lrc}, based on linearized Reed-Solomon codes \cite{linearizedRS}. We achieve download rates as in (\ref{eq camilla rate}), matching GRS storage codes (which cannot be used in our scenario since they cannot be obtained as MDS codes from puncturing MR-LRCs, thus our scheme would be \textit{less restrictive}). 

It is worth mentioning that some works not only consider the download rate, but also the upload rate (see \cite{skachek, yekhanin} and the references therein). However, in the case of our scheme, the upload cost may be considered negligible compared to the download one by folding the scheme by a large number (see the discussion before Definition \ref{def pir rate}). Thus, we do not consider the upload rate in this work. We also mention that PIR schemes for databases coded with Minimum Bandwidth Regenerating (MBR) codes were proposed in \cite{lavauzelle}. MBR codes are also used for local repair in distributed storage. However, their objective is minimizing the amount of downloaded data instead of the number of contacted nodes, as is the case for LRCs.

We remark that a PIR scheme from Gabidulin storage codes \cite{gabidulin} has been recently given in \cite{pir-networks}. Gabidulin codes can also be used to construct MR-LRCs \cite{calis}, and the PIR scheme in \cite{pir-networks} can be used for such MR-LRCs. However, the required field size would be at least $ 2^{gr} $ (see Subsection \ref{subsec pir over lin coded networks}), exponential in $ g $ and the code length $ N = gr $, in contrast with polynomial field sizes $ \max \{ r+\delta-1, g \}^r $ for our scheme (note also that $ r $ is preferably small). The main objective in \cite{pir-networks} is to give a PIR scheme where communication between the user and each server is via linear network coding \cite{linearnetwork}. We will see in Subsection \ref{subsec pir over lin coded networks} that our PIR scheme can be used in the same scenario, achieving the same rate, but with polynomial field sizes as noted above. 

The paper is organized as follows. In Section \ref{sec PIR}, we formulate general PIR schemes for MR-LRC databases. In Section \ref{sec coordinate matrix product}, we develop the mathematical tools for our PIR scheme. In Section \ref{sec main PIR scheme}, we explicitly describe our PIR scheme. Finally, in Section \ref{sec further considerations}, we discuss some further considerations.

\section{Private Information Retrieval from MR-LRC databases} \label{sec PIR}

In this section, we describe the \textit{private information retrieval} (PIR) model that we consider in this work. To the best of our knowledge, no general information-theoretical PIR model has yet been proposed for LRC databases. 

Let $ q $ be a prime power. We will denote by $ \mathbb{F} $ an arbitrary field, and by $ \mathbb{F}_q $ the finite field with $ q $ elements. Usually, we will consider $ \mathbb{F} = \mathbb{F}_{q^r} $, where $ r \geq 1 $ will be the locality of the considered storage codes. We will also denote by $ \mathbb{F}^{m \times n} $ the set of $ m \times n $ matrices with entries in $ \mathbb{F} $, and we denote $ \mathbb{F}^n = \mathbb{F}^{1 \times n} $. For a positive integer $ n $, we denote $ [n] = \{ 1,2, \ldots, n\} $. Given $ \mathcal{R} \subseteq [n] $, we denote by $ \mathbf{c}_\mathcal{R} \in \mathbb{F}^{|\mathcal{R}|} $, $ A|_\mathcal{R} \in \mathbb{F}^{m \times | \mathcal{R} |} $ and $ \mathcal{C}_\mathcal{R} \subseteq \mathbb{F}^{| \mathcal{R} |} $ the restrictions of a vector $ \mathbf{c} \in \mathbb{F}^n $, a matrix $ A \in \mathbb{F}^{m \times n} $ and a code $ \mathcal{C} \subseteq \mathbb{F}^n $, respectively, to the coordinates indexed by $ \mathcal{R} $.

Next, we recall the definitions of \textit{locally repairable codes} \cite{gopalan, kamath}.

\begin{definition}[\textbf{Locally repairable codes \cite{gopalan, kamath}}] \label{def LRC}
We say that a code $ \mathcal{C} \subseteq \mathbb{F}^n $ is a \textit{locally repairable code} (LRC) with $ (r,\delta) $ localities if there exists a partition $ [n] = \Gamma_1 \cup \Gamma_2 \cup \ldots \cup \Gamma_g $, where $ \Gamma_i \cap \Gamma_j = \varnothing $ if $ i \neq j $, such that
\begin{enumerate}
\item
$ | \Gamma_j | = r + \delta - 1 $, and
\item
$ {\rm d}_H(\mathcal{C}_{\Gamma_j}) \geq \delta $,
\end{enumerate}
for $ j = 1,2, \ldots, g $. The set $ \Gamma_j $ is called the $ j $th local group, $ r $ is called the \textit{locality}, and $ \delta $ is called the \textit{local distance}. A \textit{node} will simply be an index $ j \in [n] $.
\end{definition}

Note that $ n = (r+\delta-1)g $, but necessarily it must hold that $ k = \dim(\mathcal{C}) \leq gr $. Maximally recoverable LRCs, introduced in \cite[Def. 2.1]{blaum-RAID} and \cite[Def. 6]{gopalan-MR}, can correct all information-theoretically correctable erasure patterns for the given locality constraints $ r $, $ \delta $, $ k $ and $ g $. Such patterns are formed by any $ \delta-1 $ erasures per local group, plus any $ h = gr-k $ extra erasures anywhere else, as the following definition shows.

\begin{definition} [\textbf{Maximal recoverability \cite{blaum-RAID, gopalan-MR}}] \label{def MR}
We say that an LRC $ \mathcal{C} \subseteq \mathbb{F}^n $ with $ (r,\delta) $ localities is \textit{maximally recoverable} (MR) if, for any $ \Delta_j \subseteq \Gamma_j $ with $ | \Delta_j | = r $, for $ j = 1,2, \ldots,g $, the restricted code $ \mathcal{C}_\Delta \subseteq \mathbb{F}^{N} $ is MDS, where $ \Delta = \bigcup_{j=1}^g \Delta_j $ and $ N = | \Delta | = gr $. We say for short that $ \mathcal{C} $ is an MR-LRC. We will usually call $ \mathcal{C}_\Delta \subseteq \mathbb{F}^N $ a \textit{remaining MDS code} of $ \mathcal{C} $ (there is one for each choice of $ \Delta_j $'s).
\end{definition} 

MR-LRCs can correct more erasure patterns than most LRCs with optimal minimum distance with respect to the bound in \cite[Th. 2.1]{kamath}, such as Tamo-Barg codes \cite{tamo-barg}.

Our PIR schemes will work for the MDS codes that remain after puncturing an MR-LRC, as in Definition \ref{def MR}. As explained in Section \ref{sec intro}, the PIR schemes from \cite{hollanti} only work for GRS codes, which may not be the MDS code $ \mathcal{C}_\Delta $ that remains after puncturing an MR-LRC $ \mathcal{C} $ as in Definition \ref{def MR} since GRS have linear field sizes and MR-LRCs require super-linear field sizes (in the code length).

We will consider collusion patterns formed by unions of local groups. Note that collusion patterns strongly depend on each particular scenario. We now argue why in the LRC case it actually makes more sense to consider collusion patterns formed by unions of local groups than any set of nodes corresponding to codeword symbols:
\begin{enumerate}
\item
Communication is much more frequent and necessary among nodes inside a local group, as local correction is the most frequent type and global correction is only left to catastrophic erasure patterns. For this reason, local groups could be considered as separate storage units or even be placed geographically apart. In the extreme case $ h = gr - k = 0 $, MR-LRCs are simply Cartesian products of codes. In this case, no communication across local groups is needed and they could store completely independent data. In contrast, for repetition codes and MDS codes, communication across servers is needed to correct average erasure patterns. For these reasons, we will consider local groups $ \Gamma_j \subseteq [n] $, rather than individual nodes $ j \in [n] $, as \textit{corruptable units}. In other words, the \textit{$ j $th server} will be identified with the $ j $th local group. Thus a subset $ T \subseteq [g] $ of \textit{colluding servers} is the same as the corresponding $ |T| $ local groups and $ (r+\delta-1)|T| $ colluding nodes. 
\item
To help reduce the downloaded amount of data from the $ j $th server (i.e. $ j $th local group, see Item 1), we assume that only $ r $ stored symbols from each codeword are downloaded from that server, since the remaining $ \delta-1 $ nodes only contain locally redundant information. This means that we consider an MDS code that remains after puncturing the MR-LRC as in Definition \ref{def MR}. As explained above, GRS codes (the MDS codes considered in \cite{hollanti}) cannot come from MR-LRC after puncturing since GRS codes have linear field sizes and MR-LRCs require super-linear field sizes \cite{gopi}, hence the scheme in \cite{hollanti} does not apply to the MR-LRC scenario.
\end{enumerate}

Fix now positive integers $ b $, $ m $ and $ k \leq N = gr $, and let $ \mathbf{x}^1, $ $ \mathbf{x}^2, $ $ \ldots, $ $ \mathbf{x}^m \in \mathbb{F}_{q^r}^{b \times k} $ be the $ m $ files to be stored. Arrange them as 
$$ X = \left( \begin{array}{c}
\mathbf{x}^1 \\
\vdots \\
\mathbf{x}^m
\end{array} \right) \in \mathbb{F}_{q^r}^{bm \times k} . $$
Consider an arbitrary MR-LRC $ \mathcal{C}_{glob} \subseteq \mathbb{F}_{q^r}^n $ (the \textit{global code}) with a generator matrix of the form
\begin{equation}
G_{glob} = G_{out} \diag_g(A) \in \mathbb{F}_{q^r}^{k \times n},
\label{eq gen matrix of global code}
\end{equation}
where $ n = g(r+\delta-1) $, $ N = gr $, $ G_{out} \in \mathbb{F}_{q^r}^{k \times N} $ is a generator matrix of some $ k $-dimensional \textit{outer code} $ \mathcal{C}_{out} \subseteq \mathbb{F}_{q^r}^N $, $ A \in \mathbb{F}_q^{r \times (r + \delta - 1)} $ is a generator matrix of an $ (r+\delta-1,r) $ MDS code (the \textit{local code}), and $ \diag_g(A) = \diag(A, A, \ldots, A) \in \mathbb{F}_q^{N \times n} $ is the block-diagonal matrix with $ A $ repeated in the main block-diagonal $ g $ times.

Each file $ \mathbf{x}^i $ is then encoded into $ \mathbf{y}^i = \mathbf{x}^i G_{glob} \in \mathbb{F}_{q^r}^{b \times n} $, where $ \mathbf{y}^i_{\Gamma_j} \in \mathbb{F}_{q^r}^{b \times (r+\delta-1)} $ is stored in the $ j $th server. Let $ \Delta_j \subseteq \Gamma_j $ be the first $ r $ coordinates in $ \Gamma_j $ and assume that the first $ r $ columns of $ A $ form the identity matrix (i.e. $ A $ is systematic). Then if we disregard the nodes in $ \Gamma_j \setminus \Delta_j $, the part of the $ i $th file stored in the $ j $th server is $ \mathbf{z}^i_j \in \mathbb{F}_{q^r}^{b \times r} $, for $ j = 1,2, \ldots, g $, where
\begin{equation}
Z = XG_{out} = \left( \begin{array}{c}
\mathbf{z}^1 \\
\vdots \\
\mathbf{z}^m
\end{array} \right) = \left( \mathbf{z}_1, \mathbf{z}_2, \ldots, \mathbf{z}_g \right) = \left( \begin{array}{ccc}
\mathbf{z}_1^1 & \ldots & \mathbf{z}_g^1 \\
\vdots & \ddots & \vdots \\
\mathbf{z}_1^m & \ldots & \mathbf{z}_g^m
\end{array} \right) \in \mathbb{F}_{q^r}^{b m \times N} .
\label{eq z's if A systematic}
\end{equation}
Thus the remaining MDS code coincides with the outer code: $ \mathcal{C}_\Delta = \mathcal{C}_{out} $.

The parameter $ b $ will be called the \textit{folding parameter}. It is a common parameter that is usually considered in Coding Theory implicitly. It allows to store a larger amount of data while the encoding and decoding operations grow linearly with $ b $.

We now formalize general private information retrieval schemes for MR-LRCs.

\begin{definition} \label{def PIR}
A private information retrieval (PIR) scheme for an MR-LRC distributed storage system as described above consists, for each $ i = 1,2, \ldots, m $, of:
\begin{enumerate}
\item
Random \textit{queries} sent to the $ j $th server to retrieve the $ i $th file:
$$ \mathbf{q}_j^i = \left( \mathbf{q}_j^{i,1}, \mathbf{q}_j^{i,2}, \ldots, \mathbf{q}_j^{i,m} \right) \in \mathbb{F}_{q^r}^{b m r} , $$
where $ \mathbf{q}_j^{i, \ell} \in \mathbb{F}_{q^r}^{br} $, for $ \ell = 1, 2, \ldots, m $ and $ j = 1,2,\ldots, g $. 
\item
The corresponding \textit{response} $ \mathbf{r}_j^i = \mathbf{z}_j \cdot \mathbf{q}_j^i \in \mathbb{F}_{q^r}^r $ of the $ j $th server when requested the $ i $th file (the server only knows $ \mathbf{z}_j $ and $ \mathbf{q}_j^i $ in principle), for $ j = 1,2, \ldots, g $. The product $ \cdot $ will be given in (\ref{eq def inner product}). We will denote $ \mathbf{r}^i = (\mathbf{r}_1^i, \mathbf{r}_2^i, \ldots, \mathbf{r}_g^i) \in \mathbb{F}_{q^r}^N $.
\item
A number $ s $ of iterations of Items 1 and 2, until the $ i $th file can be recovered from the responses $ \mathbf{r}^i $ in Item 2.
\item
A reconstruction function with input the $ s $ responses $ \mathbf{r}^i $ and output the $ i $th file $ \mathbf{x}^i $.
\end{enumerate}
\end{definition}

A major difference with \cite{oneshot, hollanti} is that we do not use the usual inner product $ \mathbf{z} \cdot \mathbf{q} $, but a generalization of it (see Definition \ref{def matrix-schur products} below).

As usual in the PIR literature, our goal is to maximize the download rate, which is defined as the file size divided by the amount of downloaded data. The upload cost may be considered negligible by further \textit{folding} the scheme $ b^\prime \gg 1 $ times, thus a total of $ bb^\prime $ times. Disregarding also local redundancies, the download rate is as follows. 

\begin{definition} \label{def pir rate}
We define the \textit{download rate}, or simply rate, of a PIR scheme given as in Definition \ref{def PIR} as 
\begin{equation}
R = \frac{bk}{Ns}.
\label{eq rate}
\end{equation}
\end{definition}

We require information-theoretical privacy for a given number $ t $ of colluding servers (i.e. colluding local groups).

\begin{definition} \label{def PIR privacy}
We say that a PIR scheme as in Definition \ref{def PIR} protects against $ t $ colluding servers (i.e., colluding local groups) if, for every $ T \subseteq [g] $ of size $ t $, in each iteration of the scheme we have that
$$ I \left( \left( \mathbf{q}_j^i \right) _{j \in T} ;i \right) = 0, $$
where $ I(X; Y) $ denotes the mutual information between two random variables $ X $ and $ Y $ (see \cite[Ch. 12]{cover}).
\end{definition}

\section{Coordinate-wise and inner matrix products} \label{sec coordinate matrix product}

In this section, we define and collect the main properties of inner and coordinate-wise matrix products used in our PIR scheme (see Item 2 in Definition \ref{def PIR}). Such products will be crucial for the MR-LRCs from \cite{universal-lrc} based on linearized Reed-Solomon codes \cite{linearizedRS}. Thus we start by revisiting such codes.

\subsection{MR-LRCs based on linearized Reed-Solomon codes} \label{subsec lin RS codes}

Fix $ r \geq 1 $ and let $ \sigma : \mathbb{F}_{q^r} \longrightarrow \mathbb{F}_{q^r} $ be given by $ \sigma(a) = a^q $, for all $ a \in \mathbb{F}_{q^r} $. We next define linear operators as in \cite[Def. 20]{linearizedRS}. 

\begin{definition}[\textbf{\cite{linearizedRS}}] \label{def linearized operators}
Fix $ a \in \mathbb{F}_{q^r} $, and define its $ i $th norm as $ N_i(a) = \sigma^{i-1}(a) \cdots \sigma(a)a $, for $ i \in \mathbb{N} $. We define the $ \mathbb{F}_q $-linear operator $ \mathcal{D}_a^i : \mathbb{F}_{q^r} \longrightarrow \mathbb{F}_{q^r} $ by
\begin{equation}
\mathcal{D}_a^i(\beta) = \sigma^i(\beta) N_i(a) ,
\label{eq definition linear operator}
\end{equation}
for all $ \beta \in \mathbb{F}_{q^r} $, and all $ i \in \mathbb{N} $. Define also $ \mathcal{D}_a = \mathcal{D}_a^1 $ and observe that $ \mathcal{D}_a^{i+1} = \mathcal{D}_a \circ \mathcal{D}_a^i $, for $ i \in \mathbb{N} $. Denote by $ \mathbb{F}_{q^r} [\mathcal{D}_a] $ the polynomial ring in the operator $ \mathcal{D}_a $, for $ a \in \mathbb{F}_{q^r} $.
\end{definition}

Recall that the \textit{skew polynomial} ring $ \mathbb{F}_{q^r}[x ; \sigma] $, introduced in \cite{ore}, is the polynomial ring on the variable $ x $ but with non-commutative product given by the rule
\begin{equation}
x \beta = \sigma(\beta) x,
\label{eq defining property of skew polynomial ring}
\end{equation}
for all $ \beta \in \mathbb{F}_{q^r} $. For $ F = \sum_{i=0}^d F_i x^i \in \mathbb{F}_{q^r}[x ; \sigma] $, we define
\begin{equation}
F^{\mathcal{D}_a} = \sum_{i=0}^d F_i \mathcal{D}_a^i \in \mathbb{F}_{q^r}[\mathcal{D}_a],
\label{eq def of F^D_a}
\end{equation}
for $ a \in \mathbb{F}_{q^r} $. In the following, for $ F = \sum_{i=0}^d F_i x^i \in \mathbb{F}_{q^r}[x ; \sigma] $, for $ a \in \mathbb{F}_{q^r} $ and for $ \boldsymbol\beta = (\beta_1, \beta_2, \ldots, \beta_r) \in \mathbb{F}_{q^r}^r $, we use the notation
\begin{equation}
F^{\mathcal{D}_a}(\boldsymbol\beta) = \left( F^{\mathcal{D}_a}(\beta_1), F^{\mathcal{D}_a}(\beta_2), \ldots, F^{\mathcal{D}_a}(\beta_r) \right) \in \mathbb{F}_{q^r}^r.
\label{eq notation for vector polynomial ev}
\end{equation}
Next, given $ \mathbf{a} = (a_1, a_2, \ldots, a_g) \in \mathbb{F}_{q^r}^g $, we define the total evaluation vector of $ F $ at $ (\mathbf{a}, \boldsymbol\beta) $ as
\begin{equation}
F^{\mathcal{D}_\mathbf{a}}(\boldsymbol\beta) = \left( F^{\mathcal{D}_{a_1}}(\boldsymbol\beta), F^{\mathcal{D}_{a_2}}(\boldsymbol\beta), \ldots, F^{\mathcal{D}_{a_g}}(\boldsymbol\beta) \right) \in \mathbb{F}_{q^r}^N ,
\label{eq notation for total vector polynomial ev}
\end{equation}
where $ N = gr $. The following definition is a particular case of \cite[Def. 31]{linearizedRS}. 

\begin{definition} [\textbf{Linearized Reed-Solomon codes \cite{linearizedRS}}] \label{def lin RS codes}
Fix a primitive element $ \gamma \in \mathbb{F}_{q^r}^* $ and let $ \mathbf{a} = (\gamma^0, \gamma^1, \ldots, \gamma^{g-1}) $ $ \in \mathbb{F}_{q^r}^g $. Fix an ordered basis $ \boldsymbol\beta = (\beta_1, $ $ \beta_2, $ $ \ldots, \beta_r) $ $ \in \mathbb{F}_{q^r}^r $ of $ \mathbb{F}_{q^r} $ over $ \mathbb{F}_q $. For $ k = 0,1,2, \ldots, N $, where $ N = gr $, we define the $ (N,k) $ linearized Reed-Solomon (LRS) code as
\begin{equation*}
\mathcal{C}_{N,k}(\mathbf{a}, \boldsymbol\beta) = \left\lbrace F^{\mathcal{D}_\mathbf{a}}(\boldsymbol\beta) \in \mathbb{F}_{q^r}^N \mid F \in \mathbb{F}_{q^r}[x;\sigma] , \deg(F) < k \textrm{ or } F = 0 \right\rbrace \subseteq \mathbb{F}_{q^r}^N.
\end{equation*}
Here, the degree $ \deg(F) $ of a non-zero skew polynomial $ F = \sum_{i \in \mathbb{N}} F_i x^i \in \mathbb{F}_{q^r}[x; \sigma] $, where $ F_i \in \mathbb{F}_{q^r} $ for $ i \in \mathbb{N} $, is defined as the maximum $ i \in \mathbb{N} $ such that $ F_i \neq 0 $.
\end{definition}

Linearized Reed-Solomon codes recover Reed-Solomon codes \cite{reed-solomon} by setting $ r = 1 $ and $ \beta_1 = 1 $, and they recover Gabidulin codes \cite{gabidulin} by setting $ g = 1 $.

In \cite[Const. 1]{universal-lrc}, a construction of MR-LRCs was given based on linearized Reed-Solomon codes (Definition \ref{def lin RS codes}). This construction recovers Reed-Solomon codes if $ r = \delta = 1 $, and it recovers Cartesian products of codes if $ h = gr-k = 0 $. 

\begin{construction} [\textbf{LRS-based MR-LRC \cite{universal-lrc}}] \label{construction 1}
Fix the positive integers $ g $, $ r $ and $ \delta $, and choose any \textit{base field} size $ q > \max \{ r+\delta-3, g \} $. Next choose a dimension $ k = 1,2, \ldots, N $, where $ N = gr $, and:
\begin{enumerate}
\item
\textit{Outer code}: An $ (N,k) $ linearized Reed-Solomon code $ \mathcal{C}_{out} = \mathcal{C}_{N,k}(\mathbf{a}, \boldsymbol\beta) \subseteq \mathbb{F}_{q^r}^N $ as in Definition \ref{def lin RS codes}.
\item
\textit{Local codes}: Any linear $ (r+\delta-1, r) $ MDS code $ \mathcal{C}_{loc} \subseteq \mathbb{F}_q^{r + \delta - 1} $.
\item
\textit{Global code}: Let $ \mathcal{C}_{glob} \subseteq \mathbb{F}_{q^r}^n $, with $ n = (r + \delta - 1)g = N + (\delta - 1)g $, be given by
$$ \mathcal{C}_{glob} = \mathcal{C}_{out} \diag_g(A) \subseteq \mathbb{F}_{q^r}^n, $$
where $ A \in \mathbb{F}_q^{r \times (r + \delta - 1)} $ is any generator matrix of $ \mathcal{C}_{loc} $.
\end{enumerate}
\end{construction}

The following result is \cite[Th. 2]{universal-lrc} and states the MR and LRC properties of the global code $ \mathcal{C}_{glob} $ in Construction \ref{construction 1}.

\begin{theorem}[\textbf{\cite{universal-lrc}}] \label{th constr 1 is MR}
Let $ \mathcal{C}_{glob} \subseteq \mathbb{F}_{q^r}^n $ be the global code from Construction \ref{construction 1}, and let $ \Gamma_j \subseteq [n] $ be the subset of coordinates from $ (r+\delta-1)(j-1) + 1 $ to $ (r + \delta - 1)j $, for $ j = 1,2, \ldots, g $. Then $ \mathcal{C}_{glob} \subseteq \mathbb{F}_{q^r}^n $ has $ (r,\delta) $ localities, local groups $ \Gamma_1, \Gamma_2, \ldots, \Gamma_g $, and is maximally recoverable. Furthermore, its field size may be chosen as $ \max \{ r+\delta-3, g \}^r $.
\end{theorem}
%

\subsection{Definition and linearity properties of the products} \label{subsec definition}

Fix an ordered basis $ \boldsymbol\beta = ( \beta_{1}, $ $ \beta_{2}, $ $ \ldots, $ $ \beta_{r} ) \in \mathbb{F}_{q^r}^r $ of $ \mathbb{F}_{q^r} $ over $ \mathbb{F}_q $. Denote by $ M_{\boldsymbol\beta} : \mathbb{F}_{q^r}^r \longrightarrow \mathbb{F}_q^{r \times r} $ the corresponding \textit{matrix-representation map}, given by 
\begin{equation}
M_{\boldsymbol\beta} \left(  \mathbf{x}  \right) = \left( \begin{array}{cccc}
x^1_1 & x^1_2 & \ldots & x^1_r \\
x^2_1 & x^2_2 & \ldots & x^2_r \\
\vdots & \vdots & \ddots & \vdots \\
x^r_1 & x^r_2 & \ldots & x^r_r \\
\end{array} \right),
\label{eq def matrix representation map}
\end{equation}
for $ \mathbf{x} = (x_1, x_2, \ldots, x_r) \in \mathbb{F}_{q^r}^r $, where $ x_j^1, x_j^2, \ldots, x_j^r \in \mathbb{F}_q $ are the unique scalars such that $ x_j = \sum_{i=1}^r \beta_i x_j^i \in \mathbb{F}_{q^r} $, for $ j = 1,2, \ldots, r $. Observe that $ M_{\boldsymbol\beta} $ is an $ \mathbb{F}_q $-linear vector space isomorphism, and it is the identity map if $ r = 1 $ and $ \beta_1 = 1 $. 
	
\begin{definition} \label{def matrix-schur products}
Given $ \mathbf{x}, \mathbf{y} \in \mathbb{F}_{q^r}^r $, we define their \textit{matrix product} with respect to $ \boldsymbol\beta $ as
\begin{equation}
\mathbf{x} \star \mathbf{y} = M_{\boldsymbol\beta}^{-1}(M_{\boldsymbol\beta}(\mathbf{x}) M_{\boldsymbol\beta}(\mathbf{y})) \in \mathbb{F}_{q^r}^r .
\end{equation}
For $ N = gr $, $ \mathbf{x} = (\mathbf{x}_1, \mathbf{x}_2, \ldots, \mathbf{x}_g) \in \mathbb{F}_{q^r}^N $ and $ \mathbf{y} = (\mathbf{y}_1, $ $ \mathbf{y}_2, $ $ \ldots, $ $ \mathbf{y}_g) \in \mathbb{F}_{q^r}^N $, where $ \mathbf{x}_j, \mathbf{y}_j \in \mathbb{F}_{q^r}^r $, for $ j = 1,2, \ldots, g $, we define their \textit{coordinate-wise matrix product} as
\begin{equation}
\mathbf{x} * \mathbf{y} = ( \mathbf{x}_1 \star \mathbf{y}_1, \mathbf{x}_2 \star \mathbf{y}_2, \ldots, \mathbf{x}_g \star \mathbf{y}_g) \in \mathbb{F}_{q^r}^N ,
\label{eq def matrix-schur product}
\end{equation}
and we define their \textit{inner matrix product} $ \cdot $ as
\begin{equation}
\mathbf{x} \cdot \mathbf{y} = \sum_{j=1}^g \mathbf{x}_j \star \mathbf{y}_j \in \mathbb{F}_{q^r}^r . 
\label{eq def inner product}
\end{equation}
\end{definition}

The products $ \star $, $ * $ and $ \cdot $ all depend on the subfield $ \mathbb{F}_q \subseteq \mathbb{F}_{q^r} $ (thus $ q $ and $ r $) and the ordered basis $ \boldsymbol\beta $, but we will not denote this dependence for simplicity. The classical coordinate-wise and inner products in $ \mathbb{F}_q^N $, used in \cite{hollanti} for PIR (and in general in the literature, see \cite{oneshot}), are recovered by setting $ r = 1 $ and $ \beta_1 = 1 $ (thus $ N=g $). 

From the definitions, we note also that, if $ \mathbf{x} = (x_{1}, $ $ x_{2}, $ $ \ldots, $ $ x_{r}) \in \mathbb{F}_{q^r}^{r} $ and $ \mathbf{y} = \sum_{i=1}^r \beta_{i} \mathbf{y}^i \in \mathbb{F}_{q^r}^{r} $, with $ x_{i} \in \mathbb{F}_{q^r} $ and $ \mathbf{y}^i \in \mathbb{F}_q^{r} $, for $ i = 1,2, \ldots, r $, then
\begin{equation}
M_{\boldsymbol\beta}^{-1}(M_{\boldsymbol\beta}(\mathbf{x}) M_{\boldsymbol\beta}(\mathbf{y})) = \sum_{i=1}^r x_i \mathbf{y}^i \in \mathbb{F}_{q^r}^{r}.
\label{eq rewriting matrix-Schur product}
\end{equation}
From Equation (\ref{eq rewriting matrix-Schur product}) applied coordinate-wise, we deduce the following.

\begin{lemma} \label{lemma linearity matrix-schur product}
The coordinate-wise matrix product $ * $ is $ \mathbb{F}_q $-bilinear and $ \mathbb{F}_{q^r} $-linear in the first component, that is,
\begin{enumerate}
\item
$ (\mathbf{x} + \mathbf{x}^\prime) * \mathbf{y} = \mathbf{x}*\mathbf{y} + \mathbf{x}^\prime * \mathbf{y} $ and $ \mathbf{x}*(\mathbf{y} + \mathbf{y}^\prime) = \mathbf{x}*\mathbf{y} + \mathbf{x}*\mathbf{y}^\prime $, 
\item
$ (a \mathbf{x}) * \mathbf{y} = a (\mathbf{x} * \mathbf{y}) $ and $ \mathbf{x} * (b \mathbf{y}) = b (\mathbf{x} * \mathbf{y}) $,
\end{enumerate}
for all $ \mathbf{x}, \mathbf{x}^\prime, \mathbf{y}, \mathbf{y}^\prime \in \mathbb{F}_{q^r}^N $, all $ a \in \mathbb{F}_{q^r} $ and all $ b \in \mathbb{F}_q $.
\end{lemma}

\subsection{Products of skew and linearized polynomials} \label{subsec products skew and lin pols}

We have the following important connection between the rings $ \mathbb{F}_{q^r}[x ; \sigma] $ and $ \mathbb{F}_{q^r}[\mathcal{D}_a] $, for all $ a \in \mathbb{F}_{q^r} $. We consider $ \mathbb{F}_{q^r}[\mathcal{D}_a] $ as a ring with conventional addition and with composition of maps as multiplication, denoted by $ \circ $.

\begin{lemma} \label{lemma connection two skew pol products}
For all $ F, G \in \mathbb{F}_{q^r}[x ; \sigma] $ and all $ a \in \mathbb{F}_{q^r} $, it holds that
$$ (FG)^{\mathcal{D}_a} = F^{\mathcal{D}_a} \circ G^{\mathcal{D}_a}. $$
In particular, the map $ \mathbb{F}_{q^r}[x;\sigma] \longrightarrow \mathbb{F}_{q^r}[\mathcal{D}_a] $ given by (\ref{eq def of F^D_a}) is a (surjective) ring morphism.
\end{lemma}
\begin{proof}
Observe that
\begin{equation}
\mathcal{D}_a \circ (\beta {\rm Id}) = \sigma(\beta) \mathcal{D}_a,
\label{eq defining property product in D_a}
\end{equation}
for all $ \beta \in \mathbb{F}_{q^r} $, where $ {\rm Id} = \mathcal{D}_a^0 $ is the multiplicative identity of $ \mathbb{F}_{q^r}[\mathcal{D}_a] $, and note that (\ref{eq defining property product in D_a}) coincides with (\ref{eq defining property of skew polynomial ring}) if we set $ x = \mathcal{D}_a $. Since (\ref{eq defining property of skew polynomial ring}) is the defining property of the product in the skew polynomial ring $ \mathbb{F}_{q^r}[x ; \sigma] $, the result follows.
\end{proof}

The main result of this section is showing that products of skew polynomials become coordinate-wise matrix products after evaluation via the operators $ \mathcal{D}_a $. 

\begin{theorem} \label{th connection poly and matrix-schur products}
Let $ \boldsymbol\beta = ( \beta_{1}, \beta_{2}, \ldots, \beta_{r} ) \in \mathbb{F}_{q^r}^r $ be an ordered basis of $ \mathbb{F}_{q^r} $ over $ \mathbb{F}_q $, and let coordinate-wise matrix products $ * $ be defined via $ \boldsymbol\beta $. Then it holds that
$$ (FG)^{\mathcal{D}_\mathbf{a}}(\boldsymbol\beta) = F^{\mathcal{D}_\mathbf{a}}(\boldsymbol\beta) * G^{\mathcal{D}_\mathbf{a}}(\boldsymbol\beta), $$
for all vectors $ \mathbf{a} = (a_1, a_2, \ldots, a_g) \in \mathbb{F}_{q^r}^g $ and all skew polynomials $ F,G \in \mathbb{F}_{q^r}[x ; \sigma] $.
\end{theorem}
\begin{proof}
In \cite[Prop. 1]{RECP}, it was proven that
\begin{equation}
\sigma^\ell(\boldsymbol\beta) \star \mathbf{y} = \sigma^\ell(\mathbf{y}),
\label{eq star product for g=1}
\end{equation}
for all $ \ell \in \mathbb{N} $ and all $ \mathbf{y} \in \mathbb{F}_{q^r}^{r} $. We recall the proof of (\ref{eq star product for g=1}) for convenience of the reader. If $ \mathbf{y} = \sum_{i=1}^r \beta_{i} \mathbf{y}^i $, where $ \mathbf{y}^i \in \mathbb{F}_q^{r} $, for $ i = 1,2, \ldots, r $, then
$$ \sigma^\ell(\boldsymbol\beta) \star \mathbf{y} = \sum_{i=1}^r \sigma^\ell(\beta_{i}) \mathbf{y}^i = \sigma^\ell \left( \sum_{i=1}^r \beta_{i} \mathbf{y}^i \right) = \sigma^\ell(\mathbf{y}), $$
where the first equality is (\ref{eq rewriting matrix-Schur product}).
 
Since $ \star $ is $ \mathbb{F}_{q^r} $-linear in the first component (Lemma \ref{lemma linearity matrix-schur product}), and $ \mathcal{D}_a^\ell = N_\ell(a) \sigma^\ell $, where $ N_\ell(a) \in \mathbb{F}_{q^r} $, then we deduce from (\ref{eq star product for g=1}) that
\begin{equation*}
\mathcal{D}_a^\ell(\boldsymbol\beta) \star \mathbf{y} = (N_\ell(a) \sigma^\ell(\boldsymbol\beta)) \star \mathbf{y} = N_\ell(a) (\sigma^\ell(\boldsymbol\beta) \star \mathbf{y}) = N_\ell(a) \sigma^\ell(\mathbf{y}) = \mathcal{D}_a^\ell(\mathbf{y}),
\end{equation*}
for all $ a \in \mathbb{F}_{q^r} $, all $ \mathbf{y} \in \mathbb{F}_{q^r}^{r} $ and all $ \ell \in \mathbb{N} $. Thus the case $ g=1 $ follows by combining Lemmas \ref{lemma linearity matrix-schur product} and \ref{lemma connection two skew pol products}. 

Finally, the theorem for general $ g $ follows by applying the case $ g = 1 $ separately in each of the $ g $ coordinates over the alphabet $ \mathbb{F}_{q^r}^r $, and applying Lemma \ref{lemma connection two skew pol products}.
\end{proof}

Setting $ r = 1 $ and $ \beta_1 = 1 $, the previous theorem is nothing but the well-known fact that coordinate-wise evaluation transforms conventional polynomial products into the conventional coordinate-wise product.  

We conclude by deducing that the product of two linearized Reed-Solomon codes over the same ordered basis $ \boldsymbol\beta $ is again a linearized Reed-Solomon code. For this purpose, given $ \mathbb{F}_{q^r} $-linear codes $ \mathcal{C}_1 , \mathcal{C}_2 \subseteq \mathbb{F}_{q^r}^N $, we define their coordinate-wise matrix product as
$$ \mathcal{C}_1 * \mathcal{C}_2 = \langle \{ \mathbf{c}_1 * \mathbf{c}_2 \mid \mathbf{c}_1 \in \mathcal{C}_1, \mathbf{c}_2 \in \mathcal{C}_2 \}  \rangle \subseteq \mathbb{F}_{q^r}^N , $$
where $ \langle \mathcal{A} \rangle $ denotes the $ \mathbb{F}_q $-linear vector space generated by $ \mathcal{A} \subseteq \mathbb{F}_{q^r}^N $.

\begin{corollary} \label{cor product of lin RS}
Let $ \boldsymbol\beta = ( \beta_{1}, \beta_{2}, \ldots, \beta_{r} ) \in \mathbb{F}_{q^r}^r $ be an ordered basis of $ \mathbb{F}_{q^r} $ over $ \mathbb{F}_q $, and let coordinate-wise matrix products $ * $ be defined via $ \boldsymbol\beta $. Let also $ \mathbf{a} = (\gamma^0, \gamma^1, \ldots, \gamma^{g-1}) \in \mathbb{F}_{q^r}^g $, for a primitive element $ \gamma \in \mathbb{F}_{q^r}^* $. For any $ k_1, k_2 = 0,1,2, \ldots, N $, with $ k_1 \geq 1 $, 
$$ \mathcal{C}_{N,k_1}(\mathbf{a}, \boldsymbol\beta) * \mathcal{C}_{N,k_2}(\mathbf{a}, \boldsymbol\beta) = \mathcal{C}_{N,k_1+k_2-1}(\mathbf{a}, \boldsymbol\beta) $$
if $ k_1 + k_2 -1 \leq N $, and $ \mathcal{C}_{N,k_1}(\mathbf{a}, \boldsymbol\beta) * \mathcal{C}_{N,k_2}(\mathbf{a}, \boldsymbol\beta) = \mathbb{F}_{q^r}^N $ otherwise.
\end{corollary}
\begin{proof}
It follows by combining Lemma \ref{lemma linearity matrix-schur product}, Theorem \ref{th connection poly and matrix-schur products} and the fact that
$$ \deg(FG) = \deg(F) + \deg(G), $$
for all skew polynomials $ F,G \in \mathbb{F}_{q^r}[x ; \sigma] $.
\end{proof}

Setting $ r = 1 $ and $ \beta_1 = 1 $, Corollary \ref{cor product of lin RS} recovers the well-known fact that the classical coordinate-wise product of two Reed-Solomon codes is again a Reed-Solomon code. See for instance \cite[Prop. 3]{hollanti}. Setting $ g = 1 $, Corollary \ref{cor product of lin RS} recovers the fact that the matrix product of two Gabidulin codes is again a Gabidulin code. See for instance \cite[Lemma 10]{RECP} or \cite{pir-networks}.

\section{PIR schemes for LRS-based MR-LRC databases} \label{sec main PIR scheme}

In this section, we provide a concrete and explicit PIR scheme, as in Definition \ref{def PIR}, for the MR-LRC storage codes from Construction \ref{construction 1}. To that end, we will show how to construct the queries and how to reconstruct the file from the responses. The set of server responses, given the queries, are as described in Definition \ref{def PIR}.

Let the notation be as in Section \ref{sec PIR}, fix an ordered basis $ \boldsymbol\beta = (\beta_1, \beta_2, \ldots, \beta_r) \in \mathbb{F}_{q^r}^r $ of $ \mathbb{F}_{q^r} $ over $ \mathbb{F}_q $, and let coordinate-wise matrix products $ * $ be defined via $ \boldsymbol\beta $. We set throughout this section $ t \geq 1 $ as the target number of colluding servers, with the restriction 
\begin{equation}
k + rt \leq N,
\label{eq restriction on k , t , N}
\end{equation}
and we will set $ c = N - k - rt + 1 > 0 $. 

For clarity, we present two schemes, being the first one (Subsection \ref{subsec first scheme}) a particular case of the second one (Subsection \ref{subsec second scheme}) by setting $ b = 1 $. The first scheme is added because it is a particular case that is simpler, it is easier to understand and does not require folding by setting $ b > 1 $. However, it requires that $ k $ is divisible by $ c $. The second scheme is added simply because it is an extension of the first scheme that works for any set of parameters. 

Both schemes will achieve the PIR rate
\begin{equation}
R = \frac{c}{N} = \frac{N - k - rt + 1}{N},
\label{eq our PIR rate}
\end{equation}
coinciding with the rate of the PIR scheme in \cite{hollanti} for $ (N,k) $ GRS storage codes. However, as explained in Section \ref{sec intro}, one cannot compare these two PIR schemes, since \cite{hollanti} only works for GRS codes, which cannot be the MDS codes obtained from puncturing MR-LRCs (since GRS have linear field sizes and MR-LRCs require super-linear field sizes \cite{gopi}). Moreover, the PIR scheme from \cite{hollanti} only works for GRS codes since it uses coordinate-wise products, and GRS codes are the only MDS codes that satisfy the desired properties with respect to such products \cite{marquez}. We circumvent this limitation by considering coordinate-wise matrix products and LRS codes.

We also remark here that, mathematically speaking, the scheme in \cite{hollanti} is precisely the particular case of our second scheme by setting $ r = \delta = 1 $. Being able to extend it to arbitrary $ r $ and $ \delta $ requires a somewhat different partition of the considered vectors and matrices, and the careful use of the coordinate-wise matrix products from Section \ref{sec coordinate matrix product}. On the Coding-Theoretic side, the obvious and important difference with \cite{hollanti} is that setting $ r = \delta = 1 $ simply does not allow for local repair.

\subsection{First scheme: No folding} \label{subsec first scheme}

Our first scheme assumes no minimum folding of the files, that is, $ b = 1 $, and is precisely the particular case of our second scheme obtained by setting $ b = 1 $. Note that the stored codewords may however be further folded $ b^\prime \gg 1 $ times without folding the PIR scheme. The main disadvantage of choosing $ b=1 $ is that the dimension $ k $ must be divisible by $ c = N - k - rt + 1 $. Relaxing this divisibility assumption is the advantage of the second scheme. 

As stated above, assume that $ k = sc $, for some $ s \in \mathbb{N} $, which will be the number of iterations of the scheme. This can be trivially assumed if $ k + rt = N $, thus $ c=1 $ and $ s=k $. Note that this is the best choice for the competing parameters $ k $ and $ t $ satisfying (\ref{eq restriction on k , t , N}), but it gives the smallest PIR rate $ R = 1/N \leq c/N $ among our schemes, although $ R = 1/N $ is still far better than downloading the whole database (which gives rate $ 1/m $), since $ m \gg N $ in practice.

Fix file and iteration indices $ i = 1,2, \ldots, m $ and $ u = 1, $ $ 2, $ $ \ldots, $ $ s $, respectively. We now describe the two steps of the $ u $th iteration in Definition \ref{def PIR} to privately retrieve the $ i $th file. 

\textbf{Step 1, Queries:} Choose $ m $ codewords $ \mathbf{d}^{\ell} = (\mathbf{d}^{\ell}_1, \mathbf{d}^\ell_2, $ $ \ldots, $ $ \mathbf{d}^{\ell}_g) \in \mathbb{F}_{q^r}^{N} $ uniformly at random from $ \mathcal{C}_{N,rt}(\mathbf{a},\boldsymbol\beta) $, where $ \mathbf{d}^{\ell}_j \in \mathbb{F}_{q^r}^r $, for $ \ell = 1,2, \ldots, m $ and $ j = 1,2, \ldots, g $. The random vectors $ \mathbf{d}^{\ell} = \mathbf{d}^\ell(u) $ depend on the iteration index $ u $ (i.e. $ \mathbf{d}^\ell(1), \mathbf{d}^\ell(2), \ldots, \mathbf{d}^\ell(s) $ are identically distributed and independent), but we sometimes omit the index $ u $ for simplicity. Set
$$ \mathbf{d}_j = (\mathbf{d}_j^1, \mathbf{d}_j^2, \ldots, \mathbf{d}_j^m) \in \mathbb{F}_{q^r}^{r m}, $$
for $ j = 1,2, \ldots, g $. Define the set
$$ J_u = c(u-1) + [c], $$
where we use the notation $ a+B = \{ a+b \mid b \in B \} $, for $ a \in \mathbb{Z} $ and $ B \subseteq \mathbb{Z} $. Finally, for each server $ j = 1,2, \ldots, g $, we define its query by
\begin{equation}
\mathbf{q}_j^i(u) = \mathbf{d}_j(u) + \mathbf{e}_j^i(u) \in \mathbb{F}_{q^r}^{r m},
\label{eq def queries}
\end{equation}
where we define $ \mathbf{e}_j^i(u) \in \mathbb{F}_{q^r}^{rm} $ as being zero everywhere except in the $ i $th block of $ r $ coordinates over $ \mathbb{F}_{q^r} $, where it is defined as
\begin{equation}
M_{\boldsymbol\beta}^{-1}(I_{((j-1)r + [r]) \cap J_u}) \in \mathbb{F}_{q^r}^r .
\label{eq matrix I truncated}
\end{equation}
Here, we define $ I_J \in \mathbb{F}_q^{r \times r} $, for a set $ J \subseteq (j-1)r + [r] $, as the diagonal matrix $ I_J = \diag(\delta_1^J, \delta_2^J, \ldots, \delta_r^J) $, where $ \delta_\kappa^J = 1 $ if $ (j-1)r + \kappa \in J $, and $ \delta_\kappa^J = 0 $ otherwise. Note that $ I_\varnothing \in \mathbb{F}_q^{r \times r} $ is the zero matrix.

\textbf{Step 2, Responses:} Due to the definition of the queries in (\ref{eq def queries}) and the inner matrix product (\ref{eq def inner product}), the total response in the $ u $th iteration is 
$$ \mathbf{r}^i(u) = \sum_{\ell=1}^m (\mathbf{z}^\ell * \mathbf{d}^\ell(u)) + (\mathbf{0}_{(u-1)c} , \mathbf{z}^i_{J_u}, \mathbf{0}_{N - uc}) \in \mathbb{F}_{q^r}^N , $$
where $ \mathbf{0}_M \in \mathbb{F}_{q^r}^M $ is a zero vector of length $ M $. This is because, by the definition of the inner matrix product $ \cdot $ in $ \mathbb{F}_{q^r}^{rm} $, the response by the $ j $th server is given by 
\begin{equation*}
\begin{split}
\mathbf{r}^i_j(u) = \mathbf{z}_j \cdot \mathbf{q}^i_j(u) = \sum_{\ell = 1}^m \mathbf{z}^\ell_j \star \mathbf{q}^{i, \ell}_j(u) & = \sum_{\ell = 1}^m \left( \mathbf{z}^\ell_j \star \mathbf{d}^\ell_j(u) + \mathbf{z}^\ell_j \star \mathbf{e}^{i, \ell}_j(u) \right) \\
& = \sum_{\ell = 1}^m ( \mathbf{z}^\ell_j \star \mathbf{d}^\ell_j(u) ) + \sum_{\ell = 1}^m \mathbf{z}^\ell_j \star \mathbf{e}^{i, \ell}_j(u) \\
& = \sum_{\ell = 1}^m ( \mathbf{z}^\ell_j \star \mathbf{d}^\ell_j(u) ) + \mathbf{z}^i_j I_{((j-1)r + [r]) \cap J_u},
\end{split}
\end{equation*}
where $ \mathbf{e}^{i, \ell}_j(u) \in \mathbb{F}_{q^r}^r $ is the $ \ell $th block of $ r $ coordinates of the vector $ \mathbf{e}^i_j(u) \in \mathbb{F}_{q^r}^{rm} $, and finally,
\begin{equation*}
\begin{split}
\left( \mathbf{z}^i_j I_{((j-1)r + [r]) \cap J_u} \right)_{j = 1}^g & = \left( \mathbf{z}^i_1 I_{[r] \cap J_u}, \ldots, \mathbf{z}^i_g I_{((g-1)r + [r]) \cap J_u} \right) \\
& = (\mathbf{0}_{(u-1)c} , \mathbf{z}^i_{J_u}, \mathbf{0}_{N - uc}) \in \mathbb{F}_{q^r}^N.
\end{split}
\end{equation*}
In our view, this is the key step where we use that the $ j $th server (i.e., $ j $th corruptable unit) is the $ j $th local group after removing the local redundancies. This allows the corresponding stored data $ \mathbf{z}^\ell_j $ and query $ \mathbf{q}^\ell_j $ to be seen as $ r \times r $ matrices over $ \mathbb{F}_q $, and thus we may perform the above operations. 

\textbf{Step 3, File reconstruction:} We now describe how to recover the $ i $th file by combining the responses from all $ s $ iterations. Let $ H \in \mathbb{F}_{q^r}^{c \times N} $ be a parity-check matrix (any of them) of the linear code
$$ \mathcal{C}_{N,k + rt - 1}(\mathbf{a},\boldsymbol\beta) = \mathcal{C}_{N, k}(\mathbf{a},\boldsymbol\beta) * \mathcal{C}_{N, rt}(\mathbf{a},\boldsymbol\beta) $$ 
(recall Corollary \ref{cor product of lin RS}). For $ u = 1, $ $2, $ $ \ldots, s $, we compute
$$ \mathbf{r}^i(u) H^T = (\mathbf{0}_{(u-1)c} , \mathbf{z}^i_{J_u}, \mathbf{0}_{N - uc}) H^T, $$
which holds since
$$ \sum_{\ell=1}^m (\mathbf{z}^\ell * \mathbf{d}^\ell(u)) \in \mathcal{C}_{N,k + rt - 1}(\mathbf{a},\boldsymbol\beta). $$
Since $ \mathcal{C}_{N,k + rt - 1}(\mathbf{a},\boldsymbol\beta) $ is MDS by Theorem \ref{th constr 1 is MR}, its dual is also MDS, and we can recover the vector $ \mathbf{z}^i_{J_u} \in \mathbb{F}_{q^r}^c $ from $ \mathbf{r}^i(u) H^T $. Since we have that
$$ [k] = J_1 \cup J_2 \cup \ldots \cup J_s, $$
collecting all such $ s $ restrictions $ \mathbf{z}^i_{J_u} $, we obtain
$$ (z^i_1, z^i_2, \ldots, z^i_k) = \mathbf{x}^i (G_{out})_{[k]} \in \mathbb{F}_{q^r}^k. $$
Now, since $ \mathcal{C}_{N,k}(\mathbf{a},\boldsymbol\beta) $ is MDS, again  by Theorem \ref{th constr 1 is MR}, we may recover the $ i $th file, $ \mathbf{x}^i \in \mathbb{F}_{q^r}^k $, and we are done. 

Note that the MDS property in this last step is not necessary: If we take the generator matrix $ G_{out} $ of the outer code $ \mathcal{C}_{N,k}(\mathbf{a},\boldsymbol\beta) $ to be systematic, with the identity in the first $ k $ columns, then it simply holds that $ \mathbf{x}^i = (z^i_1, z^i_2, \ldots, z^i_k) $.

\textbf{Proof of privacy:} We now show that the proposed PIR scheme protects against any $ t $ colluding servers as in Definition \ref{def PIR privacy}. Recall from Section \ref{sec PIR} that we identify servers with local groups. Let $ T \subseteq [g] $, such that $ |T| = t $, be the set of colluding local groups. Therefore, this can be understood as an adversary gaining as information the values $ \mathbf{q}^i_j(u) \in \mathbb{F}_{q^r}^{rm} $, for $ j \in T $, and for all iterations $ u = 1,2, \ldots, s $. We will just write $ \mathbf{q}^i_j = \mathbf{q}^i_j(u) $ for simplicity. We need to prove that, for a given iteration, it holds that
$$ I((\mathbf{q}^i_j)_{j \in T}; i) = 0. $$
Since $ \mathcal{C}_{N, rt}(\mathbf{a}, \boldsymbol\beta) \subseteq \mathbb{F}_{q^r}^N $ has dimension $ rt $ and is MDS by Theorem \ref{th constr 1 is MR}, it holds that any set of $ rt $ coordinates in $ [N] $ constitute an information set for $ \mathcal{C}_{N, rt}(\mathbf{a}, \boldsymbol\beta) $. In other words, the restricted code $ \mathcal{C}_{N, rt}(\mathbf{a}, \boldsymbol\beta)_{\widetilde{T}} = \mathbb{F}_{q^r}^{rt} $ is the whole space, where $ \widetilde{T} = \bigcup_{j \in T} ((j-1)r + [r]) \subseteq [N] $ is the actual set of colluding nodes. This implies that the vectors
$$ (\mathbf{d}^{\ell}_j)_{j \in T} \in \mathbb{F}_{q^r}^{rt} $$
are uniform random variables in $ \mathbb{F}_{q^r}^{rt} $. Since the Cartesian product of independent and uniform random variables is again a uniform random variable, we deduce that
$$ (\mathbf{d}_j)_{j \in T} \in \mathbb{F}_{q^r}^{rtm} $$
is a uniform random variable in $ \mathbb{F}_{q^r}^{rtm} $. Since the vector of queries $ (\mathbf{q}_j^i)_{j \in T} $ is a translation of the random variable $ (\mathbf{d}_j)_{j \in T} $ by a deterministic vector, we deduce that $ (\mathbf{q}_j^i)_{j \in T} $ is a uniform random variable in $ \mathbb{F}_{q^r}^{rtm} $. Since there is only one uniform random variable in $ \mathbb{F}_{q^r}^{rtm} $, independently of $ i $, we deduce that $ I((\mathbf{q}^i_j)_{j \in T}; i) = 0 $, and we are done.

\subsection{Second scheme: Folding} \label{subsec second scheme}

In our second scheme, we avoid the constraint that $ k $ must be divisible by $ c = N - k - rt + 1 $. To that end, we will make use of the folding parameter $ b $ as done in \cite{hollanti}. Again, the stored codewords may be further folded $ b^\prime \gg 1 $ times without further folding the PIR scheme. We emphasize here that the scheme in \cite{hollanti} is actually recovered from this second scheme by setting $ r = \delta = 1 $, which is the case in which linearized Reed-Solomon codes recover Reed-Solomon codes (see Subsection \ref{subsec lin RS codes}). Our first scheme is also recovered from this second scheme by setting $ b = 1 $. To avoid the divisibility assumption, we define
$$ b = \frac{{\rm lcm}(c, k )}{k} \quad \textrm{and} \quad s = \frac{{\rm lcm}(c, k)}{c}, $$
hence guaranteeing that $ b k = s c $. Thus we may define
$$ h = \frac{k}{s} = \frac{c}{b}. $$
Fix file and iteration indices $ i = 1,2, \ldots, m $ and $ u = 1, $ $ 2, $ $ \ldots, $ $ s $, respectively. We now describe the two steps of $ u $th iteration in Definition \ref{def PIR} to privately retrieve the $ i $th file. 

\textbf{Step 1, Queries:} Choose $ mb $ codewords $ \mathbf{d}^{\ell, v} = (\mathbf{d}^{\ell,v}_1, $ $ \mathbf{d}^{\ell,v}_2, $ $ \ldots, $ $ \mathbf{d}^{\ell, v}_g) \in \mathbb{F}_{q^r}^{N} $, uniformly at random from $ \mathcal{C}_{N,rt}(\mathbf{a},\boldsymbol\beta) $, where $ \mathbf{d}^{\ell,v}_j \in \mathbb{F}_{q^r}^r $, for $ \ell = 1,2, \ldots, m $, $ v = 1, $ $ 2, $ $ \ldots, b $, and $ j = 1,2, \ldots, g $. As before, $ \mathbf{d}^{\ell, v} = \mathbf{d}^{\ell, v}(u) $ depends on $ u $, but we sometimes drop this in the notation. We set
$$ \mathbf{d}_j^\ell = (\mathbf{d}^{\ell,1}_j, \mathbf{d}^{\ell,2}_j, \ldots, \mathbf{d}^{\ell,b}_j) \in \mathbb{F}_{q^r}^{r b} \textrm{ and} $$
$$ \mathbf{d}_j = (\mathbf{d}_j^1, \mathbf{d}_j^2, \ldots, \mathbf{d}_j^m) \in \mathbb{F}_{q^r}^{r bm}, $$
for $ \ell = 1,2, \ldots, m $ and $ j = 1,2, \ldots, g $. Define the sets
$$ J_u^1 = h(u-1) + [h], J_u^2 = h + J_u^1, \ldots, J_u^b = h(b-1) + J_u^1. $$
Finally, for each server $ j = 1,2, \ldots, g $, we define its query by
\begin{equation}
\mathbf{q}_j^i(u) = \mathbf{d}_j(u) + \mathbf{e}_j^i(u) \in \mathbb{F}_{q^r}^{r bm}.
\label{def general query for second scheme}
\end{equation}
In this case, we define $ \mathbf{e}_j^i(u) \in \mathbb{F}_{q^r}^{rbm} $ as being zero everywhere except in the $ (b(i-1) + v) $th block of $ r $ coordinates over $ \mathbb{F}_{q^r} $, for $ v = 1,2, \ldots, b $, where it is defined as
\begin{equation}
M_{\boldsymbol\beta}^{-1}(I_{((j-1)r + [r]) \cap J_u^v}) \in \mathbb{F}_{q^r}^r .
\label{eq matrix I truncated in subset}
\end{equation}
As before, we define $ I_J \in \mathbb{F}_q^{r \times r} $, for a set $ J \subseteq (j-1)r + [r] $, as the diagonal matrix $ I_J = \diag(\delta_1^J, \delta_2^J, \ldots, \delta_r^J) $, where $ \delta_\kappa^J = 1 $ if $ (j-1)r + \kappa \in J $, and $ \delta_\kappa^J = 0 $ otherwise. As before, $ I_\varnothing \in \mathbb{F}_q^{r \times r} $ is the zero matrix.

\textbf{Step 2, Responses:} The reader can check that, from the definition of the queries in (\ref{def general query for second scheme}) and the inner matrix product (\ref{eq def inner product}), the total response in the first iteration is 
\begin{equation}
\mathbf{r}^i = \sum_{\ell=1}^m \sum_{v=1}^b (\mathbf{z}^{\ell, v} * \mathbf{d}^{\ell,v}) + (\mathbf{z}_{J_1^1}^{i,1}, \mathbf{z}_{J_1^2}^{i,2} , \ldots, \mathbf{z}_{J_1^b}^{i,b}, \mathbf{0}) \in \mathbb{F}_{q^r}^N ,
\label{eq step 2 second scheme}
\end{equation}
where $ \mathbf{0} $ has length $ N-c $. In the $ u $th iteration, the response is obtained similarly, replacing $ \mathbf{z}_{J_1^v}^{i,v} $ by $ \mathbf{z}_{J_u^v}^{i,v} $, but placed in the coordinates indexed by $ J_u^v $ taking the cyclicity of the coordinates in $ [N] $ into account, for $ v = 1,2, \ldots, b $.

\textbf{Step 3, File reconstruction:} We now describe how to recover the $ i $th file by combining the responses from all $ s $ iterations. As before, let $ H \in \mathbb{F}_{q^r}^{c \times N} $ be a parity-check matrix of 
$$ \mathcal{C}_{N,k + rt - 1}(\mathbf{a},\boldsymbol\beta) = \mathcal{C}_{N, k}(\mathbf{a},\boldsymbol\beta) * \mathcal{C}_{N, rt}(\mathbf{a},\boldsymbol\beta) $$ 
(recall Corollary \ref{cor product of lin RS}). In the first iteration, we compute
$$ \mathbf{r}^i H^T = (\mathbf{z}_{J_1^1}^{i,1}, \mathbf{z}_{J_1^2}^{i,2} , \ldots, \mathbf{z}_{J_1^b}^{i,b}, \mathbf{0}) H^T, $$
which holds since
$$ \sum_{\ell=1}^m \sum_{v=1}^b (\mathbf{z}^{\ell, v} * \mathbf{d}^{\ell,v}) \in \mathcal{C}_{N,k + rt - 1}(\mathbf{a},\boldsymbol\beta). $$
As before, $ \mathcal{C}_{N,k + rt - 1}(\mathbf{a},\boldsymbol\beta) $ is MDS by Theorem \ref{th constr 1 is MR}, and thus its dual is also MDS. Therefore we may recover $ \mathbf{z}_{J_1^1}^{i,1}, $ $ \mathbf{z}_{J_1^2}^{i,2} , $ $ \ldots, \mathbf{z}_{J_1^b}^{i,b} \in \mathbb{F}_{q^r}^h $ from $ \mathbf{r}^i H^T $. 

In a similar way, in the $ u $th iteration, we recover $ \mathbf{z}_{J_u^1}^{i,1}, $ $ \mathbf{z}_{J_u^2}^{i,2} , $ $ \ldots, \mathbf{z}_{J_u^b}^{i,b} \in \mathbb{F}_{q^r}^h $, for $ u = 1,2, \ldots, s $. For a given $ v = 1, $ $ 2, $ $ \ldots, $ $ b $, the reader can check from their definition that the sets $ J_1^v, J_2^v, \ldots, J_s^v $ are disjoint and the size of their union is $ sh = k $. Therefore, we recover $ k $ symbols of $ \mathbf{z}^{i,v} \in \mathbb{F}_{q^r}^N $ over the alphabet $ \mathbb{F}_{q^r} $, together with their indices, given by $ J_1^v \cup J_2^v \cup \ldots \cup J_s^v \subseteq [N] $. Since $ \mathcal{C}_{N,k}(\mathbf{a},\boldsymbol\beta) $ is MDS, we recover the $ v $th row of the $ i $th file, that is, $ \mathbf{x}^{i,v} \in \mathbb{F}_{q^r}^k $, for $ v = 1,2, \ldots, b $. Thus we are done by collecting all $ b $ rows, $ \mathbf{x}^{i,1}, \mathbf{x}^{i,2}, \ldots, \mathbf{x}^{i,b} $, of the $ i $th file.

\textbf{Proof of privacy:} Analogous to that in Subsection \ref{subsec first scheme}.

\subsection{Summary of parameters and complexity} \label{subsec summary}

We now discuss the complexity of the three steps of the PIR scheme in this section. We only discuss the scheme without folding, since all complexities simply get multiplied by $ b $ in the folded case. We consider the three main steps of the scheme:
\begin{enumerate}
\item
Queries: We need to generate $ sm $ uniformly random vectors in $ \mathbb{F}_{q^r}^{rt} $ and multiply each of them by a generator matrix of $ \mathcal{C}_{N, rt}(\mathbf{a}, \boldsymbol\beta) $, that is, a matrix of size $ rt \times N $. Thus this step has a complexity of $ \mathcal{O}(smrtN) = \mathcal{O}(smN^2) $ operations in $ \mathbb{F}_{q^r} $.
\item
Responses: We need to perform $ sm $ products of two matrices in $ \mathbb{F}_q^{r \times r} $ in order to compute the vectors $ \mathbf{r}^i_j = \mathbf{z}_j \cdot \mathbf{q}^i_j $, hence this step has a complexity of $ \mathcal{O}(smr^3) $ operations in $ \mathbb{F}_q $.
\item
Reconstruction: We need to compute $ s $ products of a vector in $ \mathbb{F}_{q^r}^N $ with a matrix in $ \mathbb{F}_{q^r}^{N \times c} $ in order to compute $ \mathbf{r}^i(u)H^T $, hence this step has a complexity of $ \mathcal{O}(scN) $ operations over $ \mathbb{F}_{q^r} $. Finally, computing $ \mathbf{x}^i (G_{out})_{[k]} $ is trivial if $ G_{out} $ is systematic (i.e. $ (G_{out})_{[k]} $ is the identity).
\end{enumerate}

As we can see, if the number of files $ m $ is much larger than the other parameters, then we may consider the total computational complexity as linear in $ m $ (recall that $ m \gg N $ and $ s,r,t \leq N $).

We next summarize the general PIR scheme (Subsection \ref{subsec second scheme}) in the following theorem by describing its parameters and computational complexity.

\begin{theorem} \label{th summary}
There exists a $ k $-dimensional MR-LRC $ \mathcal{C} \subseteq \mathbb{F}_{q^r}^n $ with $ (r,\delta) $ localities as in Definitions \ref{def LRC} and \ref{def MR} and, for a database coded with $ \mathcal{C} $, there exists a PIR scheme as in Definition \ref{def PIR}, with PIR rate
$$ R = \frac{N-k-rt+1}{N}, $$
where $ N = gr $ and the other parameters are as follows:
\begin{figure*}[!h]
\centering
\begin{tabular}{cc|cc}
\hline
Parameter & Restrictions & Parameter & Restrictions \\
\hline \hline
 $ r, \delta, g $ & None & Field size $ q $ & $ q > \max \{ r+\delta-3, g \} $ \\
\hline 
 No. files $ m $ & None & No. servers $ n $ & $ n = g(r+\delta-1) $ \\
\hline
 Dimension $ k $ & $ 1 \leq k \leq N $ & Colluding servers $ t $ & $ k + rt \leq N $ \\
\hline
 Iterations $ s $ & $ s = \frac{{\rm lcm}(k, N - k + rt + 1)}{N - k + rt + 1} $ & Folding $ b $ & $ b = \frac{{\rm lcm}(k, N - k + rt + 1)}{k} $ \\
\hline
\end{tabular}
\label{fig summary}
\end{figure*}
\\In addition, such a PIR scheme has a complexity of $ \mathcal{O}(smN^2) $ operations in $ \mathbb{F}_{q^r} $ for the queries, $ \mathcal{O}(smr^3) $ operations in $ \mathbb{F}_q $ for the responses, and $ \mathcal{O}(scN) $ operations over $ \mathbb{F}_{q^r} $ for the file reconstruction. If $ m $ grows while all other parameters remain constant, such complexities are linear in $ m $.
\end{theorem}

\subsection{Worked example} \label{subsec worked example}

In this subsection, we provide an example of the PIR scheme proposed in this manuscript. We will consider the simpler case of Subsection \ref{subsec first scheme} (no folding or $ b=1 $). We will consider dimension $ k = 2 $, locality $ r = 2 $, local distance $ \delta = 2 $, number of local groups $ g = 2 $ and protection against $ t = 1 $ colluding servers. The total number of nodes is $ n = g(r+\delta-1) = 6 $, but after removing $ \delta-1 = 1 $ redundant node per local groups, the number of remaining nodes is $ N = gr = 4 $. In order to consider Construction \ref{construction 1}, we choose the field size $ q = 3 > \max \{ g, r+\delta-3 \} $. Hence $ q^r = 9 $. Choosing $ s = 2 $ iterations and $ c = N -k-rt + 1 = 1 $, we notice that the hypotheses of Subsection \ref{subsec first scheme} are satisfied ($ k = sc $ and $ k+rt = N $). We keep the number of files $ m $ unrestricted (its value is not important for the example).

By considering the local codes in Construction \ref{construction 1} to be systematic, we may assume that, after removing the local redundancies, the remaining data is encoded with $ \mathcal{C}_{out} $, that is, $ \mathcal{C}_\Delta = \mathcal{C}_{out} $ as explained right after (\ref{eq z's if A systematic}). Let $ \alpha \in \mathbb{F}_9 $ be such that $ \alpha^2 = 2\alpha+1 $, which is a primitive element of $ \mathbb{F}_9 $. Let now $ a_1 = 1 $, $ a_2 = \alpha $, $ \beta_1 = 1 $ and $ \beta_2 = \alpha $. Notice that $ \beta_1 $ and $ \beta_2 $ are $ \mathbb{F}_3 $-linearly independent. Then we have
$$ G_{out} = \left( \begin{array}{cc|cc}
\beta_1 & \beta_2 & \beta_1 & \beta_2 \\
a_1 \beta_1^q & a_1 \beta_2^q & a_2 \beta_1^q & a_2 \beta_2^q
\end{array} \right) = \left( \begin{array}{cc|cc}
1 & \alpha & 1 & \alpha \\
1 & 2\alpha + 2 & \alpha & 2
\end{array} \right) \in \mathbb{F}_9^{2 \times 4} . $$
Thus if $ \mathbf{x}^i = (x^i_1, x^i_2) \in \mathbb{F}_{q^r}^k = \mathbb{F}_9^2 $ is the $ i $th file, we may consider its encoding as
$$ (z^1_1,z^1_2,z^1_3,z^1_4) = (\mathbf{z}^1_1, \mathbf{z}^1_2) = (x^i_1, x^i_2) \left( \begin{array}{cc|cc}
1 & \alpha & 1 & \alpha \\
1 & 2\alpha + 2 & \alpha & 2
\end{array} \right) = $$  
$$ \left( \begin{array}{cc|cc}
 x^i_1 + x^i_2, & \alpha x^i_1 + (2\alpha + 2) x^i_2 & x^i_1 + \alpha x^i_2, & \alpha x^i_1 + 2 x^i_2  
\end{array} \right) \in \mathbb{F}_9^4 . $$
The first server stores $ \mathbf{z}^1_1 = (x^i_1 + x^i_2, \alpha x^i_1 + (2\alpha + 2) x^i_2) \in \mathbb{F}_9^2 $ and the second server stores $ \mathbf{z}^1_2 = (x^i_1 + \alpha x^i_2, \alpha x^i_1 + 2 x^i_2) \in \mathbb{F}_9^2 $. 

We now show the three steps of the scheme for the first of the two iterations in order to recover the first file (i.e. $ i=1 $):

\textbf{Step 1, Queries:} Generate independently and uniformly at random $ \mathbf{w}^\ell_1, \mathbf{w}^\ell_2 \in \mathbb{F}_9^2 $ and, for $ j = 1,2 $, compute
$$ \mathbf{d}^\ell_j = (w^\ell_{j,1}, w^\ell_{j,2}) \left( \begin{array}{cc|cc}
1 & \alpha & 1 & \alpha \\
1 & 2\alpha + 2 & \alpha & 2
\end{array} \right) = $$  
$$ \left( \begin{array}{cc|cc}
 w^\ell_{j,1} + w^\ell_{j,2}, & \alpha w^\ell_{j,1} + (2\alpha + 2) w^\ell_{j,2} & w^\ell_{j,1} + \alpha w^\ell_{j,2}, & \alpha w^\ell_{j,1} + 2 w^\ell_{j,2}  
\end{array} \right) \in \mathbb{F}_9^4 . $$
By (\ref{eq def matrix representation map}) and (\ref{eq matrix I truncated}), we have 
$$ \mathbf{e}^1_1 = \left( M_{\boldsymbol\beta}^{-1}\left( \begin{array}{cc}
1 & 0 \\
0 & 0 
\end{array} \right), \mathbf{0} \right) = (\beta_1, 0, \ldots, 0) = (1,0, \ldots, 0) \in \mathbb{F}_9^{2m}, $$
and $ \mathbf{e}^1_2 = \mathbf{0} \in \mathbb{F}_9^{2m} $. Hence the queries for the two servers are the vectors in $ \mathbb{F}_9^{2m} $
$$ \begin{array}{rcl}
\mathbf{q}^1_1 & = & (\mathbf{d}^1_1, \mathbf{d}^2_1, \ldots, \mathbf{d}^m_1) + (1,0, \ldots, 0) , \textrm{ and} \\
\mathbf{q}^1_2 & = & (\mathbf{d}^1_2, \mathbf{d}^2_2, \ldots, \mathbf{d}^m_2),
\end{array} $$
respectively.

\textbf{Step 2, Responses:} The responses from the servers are the vectors in $ \mathbb{F}_9^2 $
$$ \begin{array}{rclcl}
\mathbf{r}^1_1 & = & \sum_{\ell = 1}^m ( \mathbf{z}^\ell_1 \star \mathbf{d}^\ell_1 ) + \mathbf{z}^1_1 \left( \begin{array}{cc}
1 & 0 \\
0 & 0 
\end{array} \right) & = & \sum_{\ell = 1}^m ( \mathbf{z}^\ell_1 \star \mathbf{d}^\ell_1 ) + (z^1_1,0) , \textrm{ and} \\
\mathbf{r}^1_2 & = & \sum_{\ell = 1}^m ( \mathbf{z}^\ell_2 \star \mathbf{d}^\ell_2 ), & & 
\end{array} $$
respectively. 

\textbf{Step 3, File reconstruction:} The vector $ \sum_{\ell = 1}^m ( \mathbf{z}^\ell_1 \star \mathbf{d}^\ell_1 ) $ is a codeword in $ \mathcal{C}_{4,3}(\mathbf{a},\boldsymbol\beta) $. Generator and parity-check matrices of such a code can be chosen, respectively, as
$$ G = \left( \begin{array}{cccc}
1 & \alpha & 1 & \alpha \\
1 & 2 \alpha + 2 & \alpha & 2 \\
1 & \alpha & 2 & 2 \alpha
\end{array} \right) \quad \textrm{and} \quad H = \left( \begin{array}{cccc}
2 \alpha & 1 & 1 & 2\alpha + 2
\end{array} \right). $$
Hence we have that
$$ \mathbf{r}^1 H^T = (\mathbf{r}^1_1, \mathbf{r}^1_2) H^T = (z^1_1,0,0,0) H^T = 2\alpha z^1_1. $$

Clearly we can recover $ z^1_1 $ from $ 2 \alpha z^1_1 $ since $ 2\alpha \neq 0 $ is known. In this way, at the end of the first iteration we have obtained $ z^1_1 $. Analogously, in the second iteration we would obtain $ z^1_2 $. In other words, at the end of the whole process we recover
$$ (z^1_1, z^1_2) = \mathbf{z}^1_1 = (x^1_1, x^1_2) \left( \begin{array}{cc}
1 & \alpha \\
1 & 2\alpha + 2
\end{array} \right), $$
and since such a matrix is invertible, we may recover the first file $ (x^1_1, x^1_2) \in \mathbb{F}_9^2 $.

\section{Further considerations} \label{sec further considerations}

\subsection{Unequal localities and local distances} \label{subsec pir unequal localities}

The results in this work may be extended, in a straightforward way, to the case where each local group $ \Gamma_j $ has a different locality $ r_j $ and local distance $ \delta_j $, for $ j = 1,2, \ldots, g $. See the next subsection for a further extension. The MR-LRC in Construction \ref{construction 1} based on linearized Reed-Solomon codes can be extended to arbitrary equal or unequal localities and local distances as long as the field is $ \mathbb{F}_{q^r} $, where $ q > g $ and $ r \geq \max \{ r_1, r_2, \ldots, r_g \} $. See \cite[Sec. III]{universal-lrc}. By choosing systematic generator matrices of the MDS local linear codes (which now are different), the remaining MDS storage code after removing all local redundancies is again a $ k $-dimensional linearized Reed-Solomon code, although of length $ N = \sum_{j=1}^g r_j $.

For $ \tau \geq 1 $ colluding nodes, the achieved rate would still be $ R = (N - k - \tau + 1)/ N $. However, $ t \geq 1 $ colluding local groups correspond in this case to a number of colluding servers that is different for different sets of local groups. In other words, the collusion pattern \cite{collusion} is generated by maximal collusion sets of different sizes. We may still proceed with the strategy in this work, that is, we may consider protecting against any $ \tau = \max \{ \sum_{j \in T} r_j \mid T \subseteq [g], |T| = t \} $ colluding nodes. Improvements on the rate $ R = (N - k - \tau + 1)/ N $ may be possible for certain cases (as in \cite[Sec. V]{collusion}), which we leave open.

Finally, the main motivation behind unequal localities and local distances is that some local groups may require faster and/or more robust repair, for instance due to hot data, while global erasure correction may be improved by considering the different localities and local distances. See \cite{chen-hao, kadhe, zeh-multiple} for more details.

\subsection{Arbitrary local linear codes and hierarchical localities} \label{subsec arbitrary local and hierarchical}

As before, the results in this work may be extended, in a straightforward way, to the case where each local group $ \Gamma_j $ uses an arbitrary $ r_j $-dimensional local linear code $ \mathcal{C}_{loc}^j \subseteq \mathbb{F}_q^{n_j} $, where $ \Gamma_j = |n_j| $ and $ r_j + \delta_j - 1 \leq n_j $, where $ \delta_j = \dd(\mathcal{C}_{loc}^j) $, for $ j = 1,2, \ldots, g $. Construction \ref{construction 1} still gives an MR-LRC for any choice of local linear codes, see \cite[Sec. IV]{universal-lrc}. Furthermore, the local codes may be dynamically, efficiently and locally updated in order to adapt to different distributed storage configurations, as discussed in \cite[Subsec. V-A]{universal-lrc}. In particular, the local codes may be in turn MR-LRCs, giving rise to multi-layer or hierarchical MR-LRCs (see \cite[Def. 7]{universal-lrc}, \cite[Subsec. V-B]{universal-lrc} and \cite[Sec. II]{MR-hierarchical}), which have optimal global distance by \cite[Th. 4]{universal-lrc}. As before, by choosing systematic generator matrices of the local codes, the remaining MDS storage code after removing the local redundancies is a linearized Reed-Solomon code of length $ N = \sum_{j=1}^g r_j $.

\subsection{PIR over linearly coded networks} \label{subsec pir over lin coded networks}

Linear network coding \cite{linearnetwork} permits maximum information flow over a network from a source to several sinks simultaneously in one shot (\textit{multicast}). In \cite{pir-networks}, PIR is considered where each server is formed by a number $ r \geq 1 $ of nodes in the database and communication between the user and each server is through a linearly coded network. 

To avoid mixing information through the network for non-colluding sets of servers, it is assumed in \cite{pir-networks} that the linearly coded networks between the user and the servers are pair-wise disjoint (after removing the user node). This makes the total transfer matrix from the user to the database and back have a block-diagonal shape $ \diag(A_1, A_2, \ldots, A_g) \in \mathbb{F}_q^{gr} $, where $ A_j \in \mathbb{F}_{q^r}^r $ is the transfer matrix from the $ j $th server to the user (we assume square transfer matrices for simplicity). In other words, the total linearly coded network from the user to the servers and back can be considered as a \textit{multishot} linearly coded network as in \cite{secure-multishot}, with one shot per server. The effect of such a channel is simply multiplying codewords by $ \diag(A_1, A_2, \ldots, A_g) $.

It was shown in \cite[Subsec. V-F]{secure-multishot} that multiplying on the right a linearized Reed-Solomon code, as in Definition \ref{def lin RS codes}, by a block-diagonal matrix $ \diag(A_1, A_2, \ldots, A_g) \in \mathbb{F}_q^{gr} $ gives again a linearized Reed-Solomon code, possibly with erasures if the matrices $ A_j $ are not full-rank. Using this fact, our PIR scheme (Subsection \ref{subsec second scheme}) may be used mutatis mutandis in the scenario described in this subsection and in \cite{pir-networks}. In the error-free and erasure-free case, the rate obtained in both works is
$$ R = \frac{N - k - rt + 1}{N}. $$
However, since Gabidulin codes \cite{gabidulin} are used in \cite{pir-networks}, the required field size is $ q_0^{gr} $, where $ q_0 \geq 2 $ is the field size of the underlying linear network code. Note that $ q_0^{gr} $ is exponential in the number of servers $ g $, whereas our scheme would still require the field size $ g^r $, which is polynomial in the number of servers $ g $.

\subsection{Systematic and non-systematic codes}

All of the results in this manuscript hold for any generator and parity-check matrix of the linear codes involved. Observe that we only need: 1) The fact that the coordinate-wise matrix product of two linearized Reed-Solomon codes is a linearized Reed-Solomon code, in Step 3 of our PIR scheme; 2) Using a parity-check matrix, systematic or not, of such a coordinate-wise matrix product of linearized Reed-Solomon codes, in Step 3 of our PIR scheme; and 3) the fact that the remaining MDS code $ \mathcal{C}_\Delta $, after removing all local redundancies $ \Gamma_j \setminus \Delta_j $, is a linearized Reed-Solomon code. To ensure the last condition, we made the assumption in Section \ref{sec PIR} that the generator matrix $ A \in \mathbb{F}_q^{r \times (r+ \delta -1)} $ is systematic, having its first $ r $ columns equal to those of the identity matrix. However, this assumption can be easily lifted. This is because, if $ A_r \in \mathbb{F}_q^{r \times r} $ is formed by the first $ r $ columns of the matrix $ A $, whether $ A_r $ is the identity matrix or not, it holds that
$$ \mathcal{C}_{N,k}(\mathbf{a}, \boldsymbol\beta) A_r = \mathcal{C}_{N,k}(\mathbf{a}, \boldsymbol\beta A_r), $$
with notation as in Definition \ref{def lin RS codes}, where $ \boldsymbol\beta A_r \in \mathbb{F}_{q^r}^r $ is another ordered basis of $ \mathbb{F}_{q^r} $ over $ \mathbb{F}_q $ (see also \cite[Subsec. V-F]{secure-multishot}). Thus our PIR scheme still works in this case, simply by replacing $ \boldsymbol\beta $ by $ \boldsymbol\beta A_r $.

\section*{Acknowledgement}

The author gratefully acknowledges the support from The Independent Research Fund Denmark (Grant No. DFF-7027-00053B). The author also wishes to thank the anonymous reviewers, who helped improve the presentation of the manuscript.

\end{document}